\documentclass[aps,pre,twocolumn,amsmath,amssymb,floatfix]{revtex4-2}

\usepackage{graphicx}
\usepackage{bm}
\usepackage{booktabs}
\usepackage{amsthm}
\usepackage[colorlinks=true,linkcolor=blue,citecolor=blue,urlcolor=blue]{hyperref}

\newtheorem{proposition}{Proposition}
\newtheorem{theorem}{Theorem}

\newcommand{\Az}{\mathrm{AZ}}
\newcommand{\Z}{\mathbb{Z}}

\begin{document}

\title{Parity-dependent sublattice structure of local $2\times2$ correlations
in domino tilings of the Aztec diamond}

\author{Chaeyeon Kim}
\author{Suhwan Yoo}
\author{Jeehoon Choi}
\author{Hee Young Kwon}
\author{Changyeon Won}
\email{cywon@khu.ac.kr}
\affiliation{Department of Physics, Kyung Hee University, Seoul 02447, Republic of Korea}

\date{\today}

\begin{abstract}
We study local $2\times2$ correlations in uniformly random domino tilings of the
Aztec diamond $\Az(n)$, a dimer model with frozen corners and a liquid region
bounded by the Arctic circle, and compare them with the square chessboard.  For
every elementary square the probabilities $P_H$ and $P_V$ of the two coverings
by a pair of parallel dominoes satisfy $P_H=P_V=P_1P_2+P_3P_4$, where
$P_1,\dots,P_4$ are the probabilities of the four surrounding edges.  This
known identity, a case of Kuo's graphical condensation, reduces the spatial
profile of $P_H$ to single-edge probabilities, which we compute exactly from
the inverse Kasteleyn matrix for all squares of $\Az(n)$.  The squares form two
interleaved sublattices, distinguished by the color of their lower-left cell.
The parity of $n$ decides which sublattice dominates near the center: the one
containing the central square for odd $n$, and the other one for even $n$.
This sign rule holds for the averaged sublattice contrast for all
$4\le n\le22$.  It also fixes the location of the global maximum of $P_H$,
which lies at or next to the center for odd $n$ and several lattice spacings
away for even $n$.  Within the central sublattice, the
preferred square follows a mod-4 pattern that matches the exactly pinned mean
height of the center.  On the $L\times L$ chessboard the same rule holds with $n$
replaced by the half-width $L/2$, but the contrast is two to six times weaker,
the central edge probability never equals $1/4$, and the maxima of $P_H$ are
pinned to the corners.  The exact $1/4$ values and the interior,
parity-selected maxima are therefore features of the Aztec geometry.
\end{abstract}

\maketitle

\section{Introduction}\label{sec:intro}

Dimer models on bipartite lattices are among the few exactly solvable models of
classical statistical mechanics~\cite{Kasteleyn61,TF61,Fisher61}.  Dimer
configurations describe the dense packing of diatomic molecules adsorbed on a
lattice~\cite{FR37}.  They also arise as the constrained low-energy manifolds of
frustrated magnets~\cite{Henley10,CMS12}, including the resonating-valence-bond
description of quantum magnets~\cite{RK88}.  In artificial spin
ice~\cite{Wang06,NMS13}, vertex-frustrated geometries such as the Shakti
lattice~\cite{MNN13,CMN13,Gilbert14} have low-energy states that map onto dimer
covers of the square lattice, and their kinetics have been imaged
directly~\cite{Lao18}.  In all these settings, local dimer arrangements encode
the correlation functions of the system~\cite{FS63}.

Domino tilings of the Aztec diamond $\Az(n)$ are a dimer model on a finite
region of the square lattice.  The region has $2^{n(n+1)/2}$
tilings~\cite{EKLP92a,EKLP92b}, and exact samples are generated by domino
shuffling~\cite{EKLP92a,Propp03}.  Its boundary conditions induce a macroscopic
phase separation~\cite{JPS95,CKP01,KOS06}.  Near the four corners the tiling is
frozen into a brick pattern, while inside the Arctic circle it is a fluctuating
liquid.  The fluctuations of the interface are governed by the Airy
process~\cite{Johansson05}, and the interior is described by an inhomogeneous
Gaussian free field~\cite{Kenyon01,ADSV16}.  Local statistics have been studied
asymptotically~\cite{CEP96,CJY15}.  At finite size, the probability of any
local pattern is a determinant built from the inverse Kasteleyn
matrix~\cite{Kasteleyn61,Kenyon97,KenyonLectures}, and for the Aztec diamond
Helfgott gave the relevant coupling function explicitly in terms of Krawtchouk
polynomials~\cite{Helfgott98}.

This paper examines local $2\times2$ correlations in this setting.  An
elementary $2\times2$ square admits two coverings by a pair of parallel
dominoes, horizontal ($H$) or vertical ($V$).  Several exact facts about them
are known; they are collected in Sec.~\ref{sec:known}.  The identity
$P_H=P_V=P_1P_2+P_3P_4$ holds on every four-cycle face of every planar
bipartite graph~\cite{Propp03,Kuo04}.  The central edge probability has an exact
mod-4 structure~\cite{Schonfelder25}.  The central $2\times2$ probabilities, and
the algebraic form of $P_H$ at an arbitrary position, follow from counts of
Aztec diamonds with a $2\times2$ hole~\cite{Ciucu97,Schonfelder25}.  What these
results do not address is the spatial structure of $P_H$: where it is large,
and how this depends on $n$.

Our results, presented in Sec.~\ref{sec:results}, are numerical and are of
three kinds.  First, the elementary squares form two interleaved sublattices.
Near the center the sublattice that contains the central square dominates for
odd $n$ and the other one for even $n$, for every $n$ computed.  Second, this
sign rule fixes the location of the global maximum of $P_H$, which lies at or
next to the center for odd $n$ and several lattice spacings away for even $n$.
Third, a comparison with the $L\times L$ chessboard shows that the same rule
holds there with $n$ replaced by the half-width $L/2$.  On the chessboard,
however, the contrast is much weaker, no central value equals $1/4$, and the
maxima of $P_H$ sit at the corners.  Section~\ref{sec:discussion} discusses
what is specific to the Aztec geometry and gives a height-function reading of
the mod-4 pattern near the center; Sec.~\ref{sec:conclusion} summarizes.

\section{Background and known results}\label{sec:known}

\subsection{Aztec diamond and Kasteleyn matrix}\label{subsec:kast}

The Aztec diamond of order $n$ is the set of cells
\begin{equation}
\Az(n)=\left\{(i,j)\in\Z^2:\ \left|i+\tfrac12\right|+\left|j+\tfrac12\right|\le n\right\},
\end{equation}
where the cell $(i,j)$ is the unit square $[j,j+1]\times[i,i+1]$; thus $i$
labels rows (increasing upward) and $j$ labels columns.  $\Az(n)$ consists of
$2n(n+1)$ cells.  We color $(i,j)$ black if $i+j$ is even and white otherwise.
A domino is an edge joining adjacent black and white cells, and a domino tiling
is a perfect matching of the cell-adjacency graph [Fig.~\ref{fig:arctic}(a)].
Figure~\ref{fig:arctic}(b) shows a uniformly random tiling of a large diamond,
and Fig.~\ref{fig:arctic}(c) the frozen/liquid structure obtained from the
edge probabilities introduced below.

\begin{figure*}[t]
  \centering
  \includegraphics[width=\textwidth]{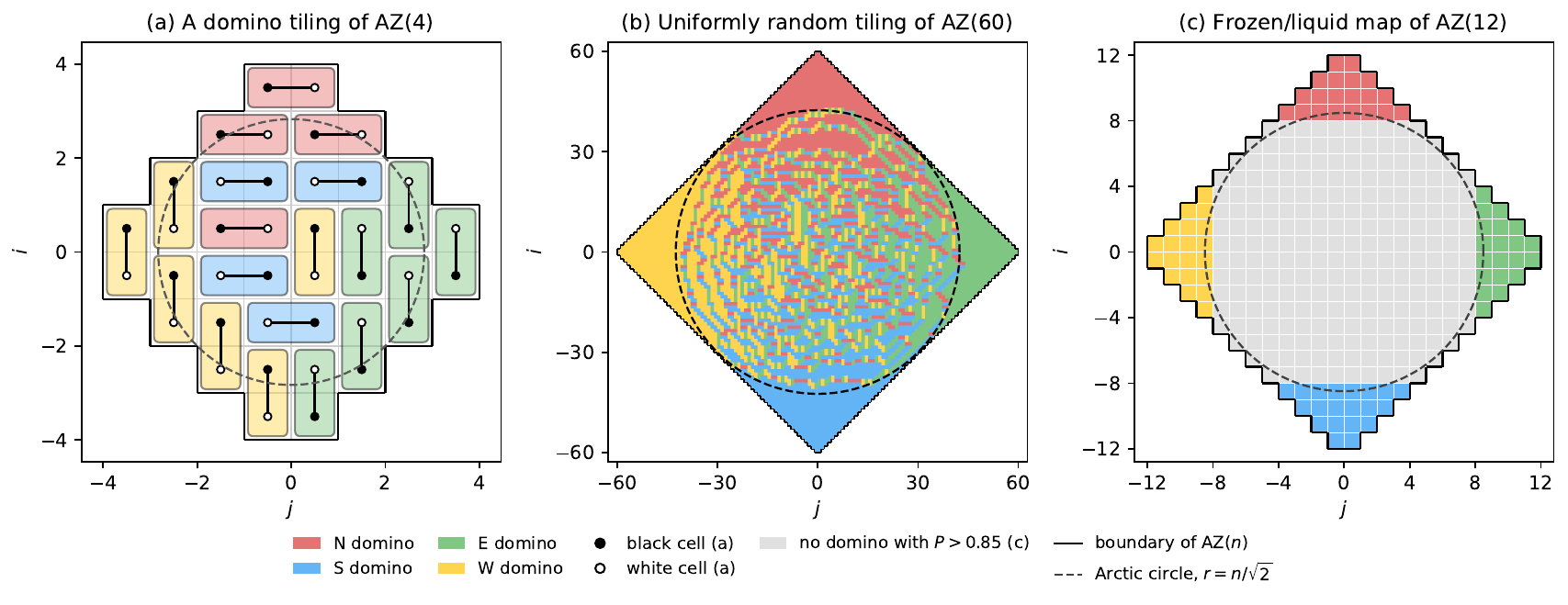}
  \caption{Domino tilings of the Aztec diamond.
    (a)~One domino tiling of $\Az(4)$.  Filled and open dots mark black and
    white cells; each domino is drawn with the edge joining its two cells, so
    that the tiling is a perfect matching of the cell-adjacency graph.
    (b)~A uniformly random tiling of $\Az(60)$, sampled exactly by domino
    shuffling~\cite{EKLP92a,Propp03}.  The four frozen corners, each filled by
    a single domino type, meet the liquid region along the Arctic
    circle~\cite{JPS95}.
    (c)~Frozen/liquid map of $\Az(12)$ from the inverse Kasteleyn matrix.  Each
    cell is colored by the type of its most likely domino when that probability
    exceeds $0.85$, and gray otherwise.
    A horizontal domino whose left cell $(i,j)$ has $i+j+n$ even (odd) is of
    type N (S), and a vertical domino whose lower cell has $i+j+n$ even (odd)
    is of type E (W).  In all panels the staircase line is the boundary of
    $\Az(n)$ and the dashed circle is the Arctic circle of radius $n/\sqrt2$.}
  \label{fig:arctic}
\end{figure*}

We index the Kasteleyn matrix $K$ by black cells (rows) and white cells
(columns) and use the phases
\begin{equation}\label{eq:K}
  K[b,w]=
  \begin{cases}
    +1, & w=b+(0,1)\ \text{(right)},\\
    -1, & w=b+(0,-1)\ \text{(left)},\\
    +i, & w=b+(1,0)\ \text{(up)},\\
    -i, & w=b+(-1,0)\ \text{(down)},
  \end{cases}
\end{equation}
with $K[b,w]=0$ for nonadjacent cells.  For the two black cells $b_1,b_2$ and
the two white cells $w_1,w_2$ of every elementary square, labeled as in
Sec.~\ref{subsec:local}, the phases satisfy the Kasteleyn condition
\begin{equation}\label{eq:kasteleyn-sign}
  K[b_1,w_1]K[b_2,w_2]=-K[b_1,w_2]K[b_2,w_1].
\end{equation}
The single-edge and two-edge probabilities are then~\cite{Kasteleyn61,Kenyon97}
\begin{gather}
  P(b\leftrightarrow w)=K[b,w]\,K^{-1}[w,b], \label{eq:single}\\
  \frac{P(e_1\wedge e_2)}{K[b_1,w_1]K[b_2,w_2]}
  =\det\begin{pmatrix}
   K^{-1}_{w_1b_1} & K^{-1}_{w_1b_2}\\
   K^{-1}_{w_2b_1} & K^{-1}_{w_2b_2}
  \end{pmatrix}, \label{eq:dpp2}
\end{gather}
where $e_r$ denotes the edge $b_r\leftrightarrow w_r$ and $K^{-1}_{wb}\equiv K^{-1}[w,b]$.

\subsection{Elementary squares and the local identity}\label{subsec:local}

The elementary square with corner cell $(a,b)$ is $\{a,a+1\}\times\{b,b+1\}$.
We call it black-anchored if $a+b$ is even and white-anchored if $a+b$ is odd.
In a black-anchored square we write $b_1=(a,b)$, $w_1=(a,b+1)$, $w_2=(a+1,b)$
and $b_2=(a+1,b+1)$, so that $b_1w_1$ and $b_2w_2$ are the horizontal edges and
$b_1w_2$ and $b_2w_1$ the vertical ones [Fig.~\ref{fig:background}(b)]; in a
white-anchored square the cells are labeled so that the same holds.  We write
$P_1=P(b_1\leftrightarrow w_1)$, $P_2=P(b_2\leftrightarrow w_2)$,
$P_3=P(b_1\leftrightarrow w_2)$ and $P_4=P(b_2\leftrightarrow w_1)$, and $P_H$
($P_V$) for the probability that both horizontal (vertical) edges occur.  The
class of a square is preserved by all symmetries of $\Az(n)$, and the central
square, with corner cell $(-1,-1)$, is black-anchored.  The centers of the
black-anchored squares, and those of the white-anchored squares, each form a
square lattice of spacing $\sqrt2$ rotated by $45^\circ$, and the two lattices
interleave [Fig.~\ref{fig:background}(a)].

\begin{figure*}[t]
  \centering
  \includegraphics[width=0.86\textwidth]{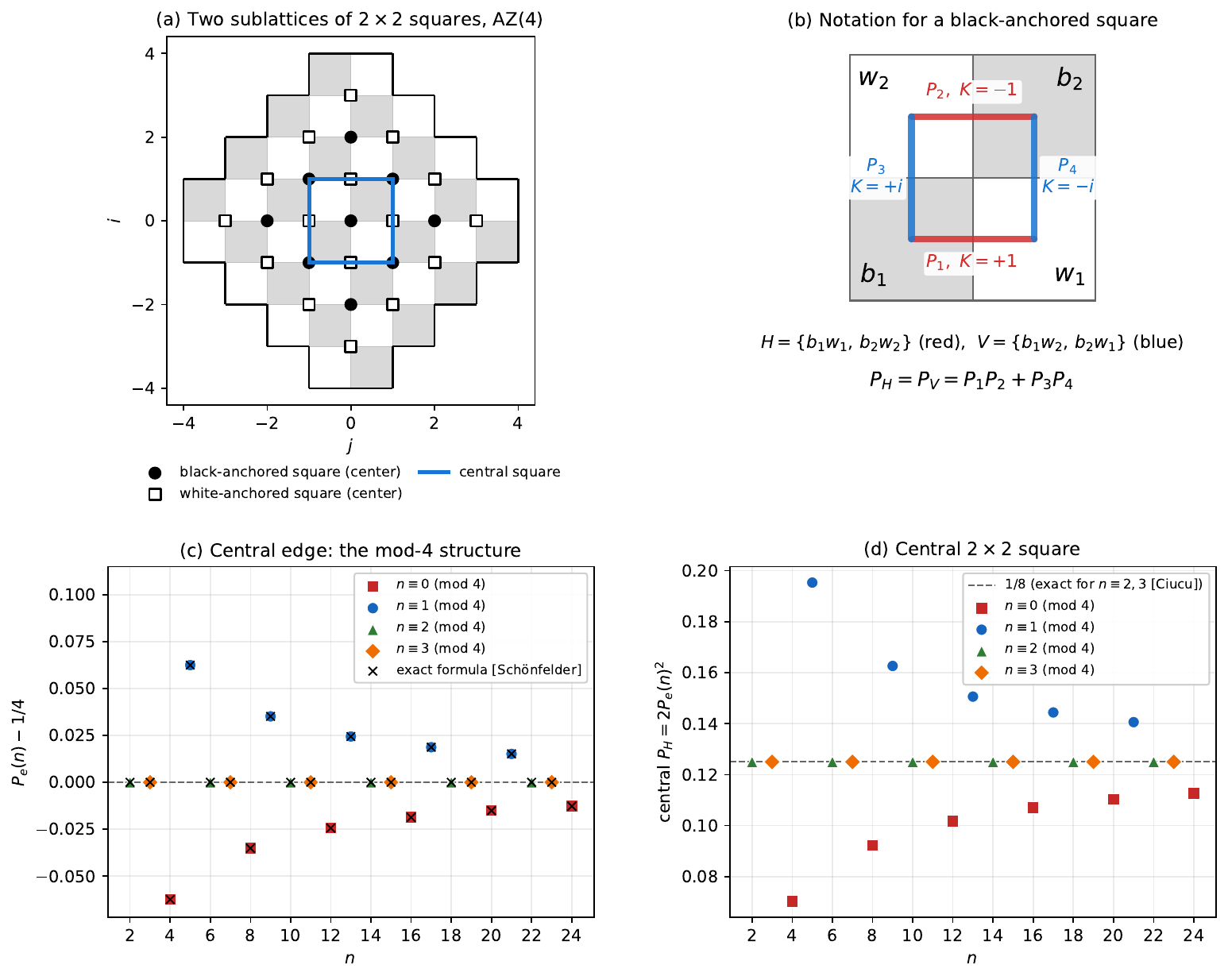}
  \caption{Notation and known results.
    (a)~The two sublattices of elementary squares in $\Az(4)$.  The center of
    each $2\times2$ square is marked by a filled circle (black-anchored) or an
    open square (white-anchored); gray cells are black.  The blue outline is
    the central square.
    (b)~Notation for a black-anchored square: edge probabilities, Kasteleyn
    phases~\eqref{eq:K}, horizontal pair $H$ (red) and vertical pair $V$
    (blue).
    (c)~Central edge probability $P_e(n)-1/4$ for $2\le n\le24$ from the
    inverse Kasteleyn matrix, grouped by $n\bmod4$, with the exact formula of
    Theorem~\ref{thm:central-edge} (crosses).
    (d)~Central $P_H=2P_e(n)^2$.  It equals $1/8$ for
    $n\equiv2,3\pmod4$~\cite{Ciucu97} and lies above (below) $1/8$ for
    $n\equiv1$ ($n\equiv0$) $\pmod4$.}
  \label{fig:background}
\end{figure*}

\begin{proposition}[Propp~\cite{Propp03}, Kuo~\cite{Kuo04}]\label{thm:local}
For every elementary square contained in $\Az(n)$,
\begin{equation}\label{eq:local}
  P_H=P_V=P_1P_2+P_3P_4.
\end{equation}
\end{proposition}

\begin{proof}
Replacing the two horizontal dominoes of the square by two vertical ones is a
bijection between the tilings containing $H$ and those containing $V$, so
$P_H=P_V$.  For the product formula, let $\kappa_{rs}=K[b_r,w_s]$.  Since
$|\kappa_{rs}|=1$, Eq.~\eqref{eq:single} gives
$K^{-1}[w_s,b_r]=P(b_r\leftrightarrow w_s)/\kappa_{rs}$, and Eq.~\eqref{eq:dpp2}
becomes
\begin{equation*}
P_H=P_1P_2-\frac{\kappa_{11}\kappa_{22}}{\kappa_{12}\kappa_{21}}\,P_3P_4
=P_1P_2+P_3P_4
\end{equation*}
by Eq.~\eqref{eq:kasteleyn-sign}.
\end{proof}

Equivalently, Eq.~\eqref{eq:local} is Corollary~2.2 of Ref.~\cite{Kuo04},
obtained by dividing Kuo's condensation identity
\begin{multline}
M(G)\,M(G-\{a,b,c,d\})=M(G-\{a,b\})\,M(G-\{c,d\})\\
+M(G-\{a,d\})\,M(G-\{b,c\})
\end{multline}
by $M(G)^2$, where $M(\cdot)$ counts perfect matchings.  The identity is local
and holds at finite size on every four-cycle face of every planar bipartite
graph with uniform weights; we confirmed it by exact enumeration, independently
of any Kasteleyn matrix, on chessboards, rectangles, a board with a $2\times2$
hole, and irregular regions.  With edge weights $\nu$ it reads
$P_H=P_1P_2+(\nu_{11}\nu_{22}/\nu_{12}\nu_{21})P_3P_4$.  In terms of edge
indicators it states
$\operatorname{Cov}(\mathbf 1_{b_1w_1},\mathbf 1_{b_2w_2})=P_3P_4\ge0$: opposite
edges of a square are never anticorrelated.  For the computations below its
main use is practical, since it gives $P_H$ on every square from single-edge
probabilities.

\subsection{Central values}\label{subsec:central}

The four edges of the central square, $b_1=(-1,-1)$, $w_1=(-1,0)$,
$w_2=(0,-1)$, $b_2=(0,0)$, are related by the $90^\circ$ rotation of $\Az(n)$
and have a common probability $P_e(n)$.  Hence Eq.~\eqref{eq:local} gives, for
the central square,
\begin{equation}\label{eq:center}
  P_H=P_V=2P_e(n)^2 .
\end{equation}

\begin{theorem}[Sch\"onfelder~\cite{Schonfelder25}]\label{thm:central-edge}
For $p\ge1$,
\begin{equation}\label{eq:central-exact}
P_e(n)=
\begin{cases}
\frac14-2^{-4p}\binom{2p-1}{p}^{2}, & n=4p,\\[3pt]
\frac14+2^{-4p}\binom{2p-1}{p}^{2}, & n=4p+1,\\[3pt]
\frac14, & n=4p+2,\,4p+3.
\end{cases}
\end{equation}
\end{theorem}

This is the central case of Sch\"onfelder's $1/4$ phenomenon: every fixed
placement probability equals $1/4$ plus a position-dependent rational function
times a factor depending only on $n$~\cite{Schonfelder25}.  In particular
$P_e(4p)+P_e(4p+1)=\frac12$, and $|P_e(n)-\frac14|\sim1/(\pi n)$ for
$n\equiv0,1\pmod4$.  Figure~\ref{fig:background}(c) compares our Kasteleyn
computations with Eq.~\eqref{eq:central-exact} for $2\le n\le24$ (maximal
deviation $1.3\times10^{-13}$), and Fig.~\ref{fig:background}(d) shows the
central $P_H$.  The central square is thus special in two ways: for
$n\equiv2,3\pmod4$ its value is exactly the bulk value $1/8$, and for
$n\equiv0,1\pmod4$ it deviates from $1/8$ in opposite directions.
Section~\ref{sec:results} asks how this central behavior fits into the spatial
profile of $P_H$.  Table~\ref{tab:central} lists exact central values.

\begin{table}[b]
\caption{Central edge and central $2\times2$ probabilities from
Eq.~\eqref{eq:center} and Theorem~\ref{thm:central-edge}.  The $P_H$ column
agrees with the holed-diamond counts of Refs.~\cite{Ciucu97,Schonfelder25}.}
\label{tab:central}
\begin{ruledtabular}
\begin{tabular}{cccc}
$n$ & $n\bmod4$ & $P_e(n)$ & $P_H=P_V$\\
\hline
4  & 0 & $3/16$   & $9/128$\\
5  & 1 & $5/16$   & $25/128$\\
6, 7  & 2, 3 & $1/4$ & $1/8$\\
8  & 0 & $55/256$ & $3025/32768$\\
9  & 1 & $73/256$ & $5329/32768$\\
12 & 0 & $231/1024$ & $53361/524288$\\
13 & 1 & $281/1024$ & $78961/524288$\\
\end{tabular}
\end{ruledtabular}
\end{table}

Since $P_H=M(\Az(n)\setminus\square)/M(\Az(n))$, where $\square$ is the central
$2\times2$ square, these values are equivalent to counts of Aztec diamonds with
a central $2\times2$ hole.  For $n=4p+2$ and $4p+3$ the counts are
$2^{8p^2+10p}$ and $2^{8p^2+14p+3}$~\cite{Ciucu97}, which give $P_H=1/8$.  For
$n\equiv0,1\pmod4$ they are given in Corollary~6.2 of
Ref.~\cite{Schonfelder25}.  For $n=4p$, the middle term of the expression
printed there must carry a minus sign,
\begin{equation}
P_H=\tfrac18-2^{-4p}\tbinom{2p-1}{p}^2+2^{1-8p}\tbinom{2p-1}{p}^4
\qquad (n=4p),
\end{equation}
equivalently $g_{0,0,0}=-4$ instead of $4$ in its proof.  A direct count for
$n=4$ gives $72$ tilings of the holed diamond ($P_H=9/128$), in agreement with
the corrected expression and with Eq.~\eqref{eq:center}; the printed expression
would give $200$ tilings ($P_H=25/128$), a value attained by no $2\times2$
square of $\Az(4)$ (Table~\ref{tab:maxloc}).  More generally, a $2\times2$ hole
at an arbitrary position gives~\cite{Schonfelder25}
\begin{multline}\label{eq:schoenfelder-form}
P_H=\tfrac18+2^{-4p-\alpha-2}\tbinom{2p-1}{p}^2g_1(p)\\
      +2^{-8p-2\alpha}\tbinom{2p-1}{p}^4g_2(p),\qquad n=4p+\alpha,
\end{multline}
with rational functions $g_1$ and $g_2$ that depend on the position.

\section{Results: spatial structure of $P_H$}\label{sec:results}

\subsection{Method}\label{subsec:method}

For each $n$ we construct $K$ from Eq.~\eqref{eq:K}, invert it numerically, and
obtain all single-edge probabilities from Eq.~\eqref{eq:single}.  $P_H$ is then
evaluated on every elementary square by Eq.~\eqref{eq:local}.  As a check, we
evaluated the two-edge determinant~\eqref{eq:dpp2} directly on every square for
$2\le n\le17$; the largest discrepancy was below $10^{-16}$.  The central values
in Table~\ref{tab:central} are exact; all other numbers in this section are
double-precision results.  For chessboards we use the same phases
(Sec.~\ref{subsec:chess}).

\subsection{Spatial profile in $\Az(12)$}\label{subsec:az12}

Figure~\ref{fig:spatial-ph}(a) shows $P_H$ on all $265$ elementary squares of
$\Az(12)$.  Each value is drawn as a unit square centered on its $2\times2$
square, so that the $121$ black-anchored and $144$ white-anchored values do not
overlap.  In the frozen corners $P_H$ is close to zero, as expected in a brick
pattern, and it rises across the Arctic circle.  Inside the liquid region the
profile is not smooth: black- and white-anchored squares form an interleaved
checkerboard with visibly different values.  The central square has
$P_H=0.10178$.  The largest black-anchored value, $0.12381$, occurs on the four
squares diagonally adjacent to it.  The global maximum, $0.14254$, occurs on
eight white-anchored squares at distance $\sqrt{17}\approx4.1$ from the center.

Figure~\ref{fig:spatial-ph}(b) shows the underlying single-edge probabilities.
The frozen corners appear as rows of saturated favored (red) and disfavored
(blue) edges.  Inside the Arctic circle the deviations are smaller and form a
pattern of period two, which produces the sublattice checkerboard in panel (a).

\begin{figure*}[t]
  \centering
  \includegraphics[width=0.9\textwidth]{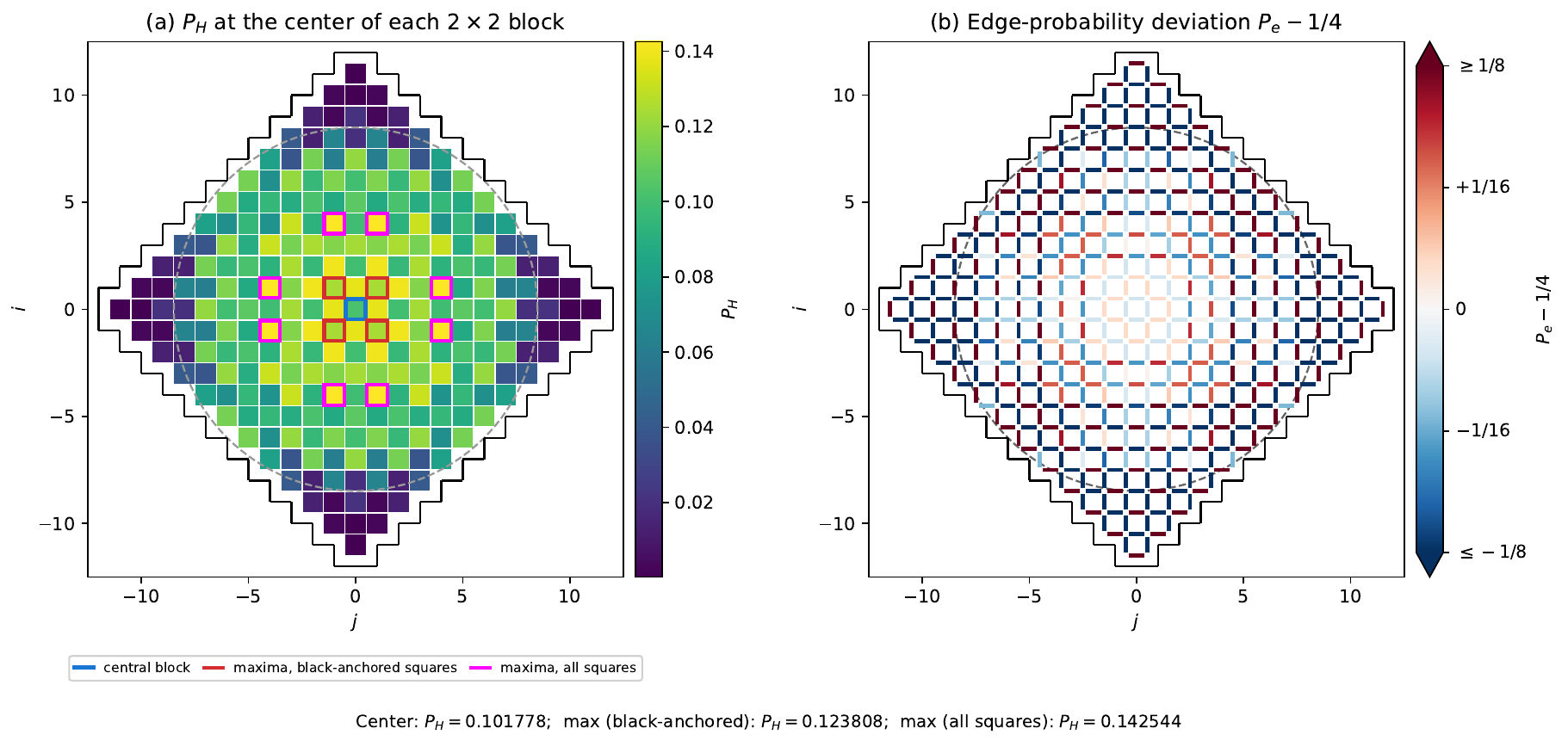}
  \caption{Local probabilities in $\Az(12)$.
    (a)~$P_H$ for every elementary square, drawn as a unit square centered on
    the corresponding $2\times2$ square.  Blue: central square.  Red: the four
    maxima among black-anchored squares.  Magenta: the eight global maxima,
    which are white-anchored.
    (b)~Deviation $P_e-1/4$ of every single-edge probability.  Each edge is
    drawn as a segment joining the centers of its two cells and colored blue
    ($P_e<1/4$), white ($P_e=1/4$), or red ($P_e>1/4$), saturating at $\pm1/8$.
    The staircase line is the boundary of $\Az(12)$ and the dashed circle the
    Arctic circle.}
  \label{fig:spatial-ph}
\end{figure*}

\subsection{Dependence on $n$: a parity rule}\label{subsec:parity}

Figure~\ref{fig:residue} shows the same profile for $\Az(9)$, $\Az(10)$,
$\Az(11)$ and $\Az(12)$, one diamond from each residue class of $n\bmod4$.  The
interleaved checkerboard is present in each case, but the brighter sublattice
near the center alternates.  In $\Az(9)$ and $\Az(11)$ the black-anchored
squares, including the central one, are brighter; in $\Az(10)$ and $\Az(12)$
the white-anchored squares are.

To quantify this, we define the sublattice contrast $\Delta(n)$ as the mean of
$P_H$ over the black-anchored squares whose centers lie within distance $3$ of
the center of $\Az(n)$, minus the same mean over white-anchored squares.
Table~\ref{tab:contrast} shows that
\begin{equation}\label{eq:parity-rule}
\operatorname{sgn}\Delta(n)=(-1)^{n+1}\qquad (4\le n\le22).
\end{equation}
The same signs are obtained when the radius $3$ is replaced by $n/(2\sqrt2)$,
half the radius of the Arctic circle.  The magnitude of $\Delta(n)$ decreases
slowly with $n$, from $0.06$--$0.08$ at $n=4,5$ to $0.012$--$0.018$ at
$n=21,22$.

\begin{figure*}[t]
  \centering
  \includegraphics[width=0.78\textwidth]{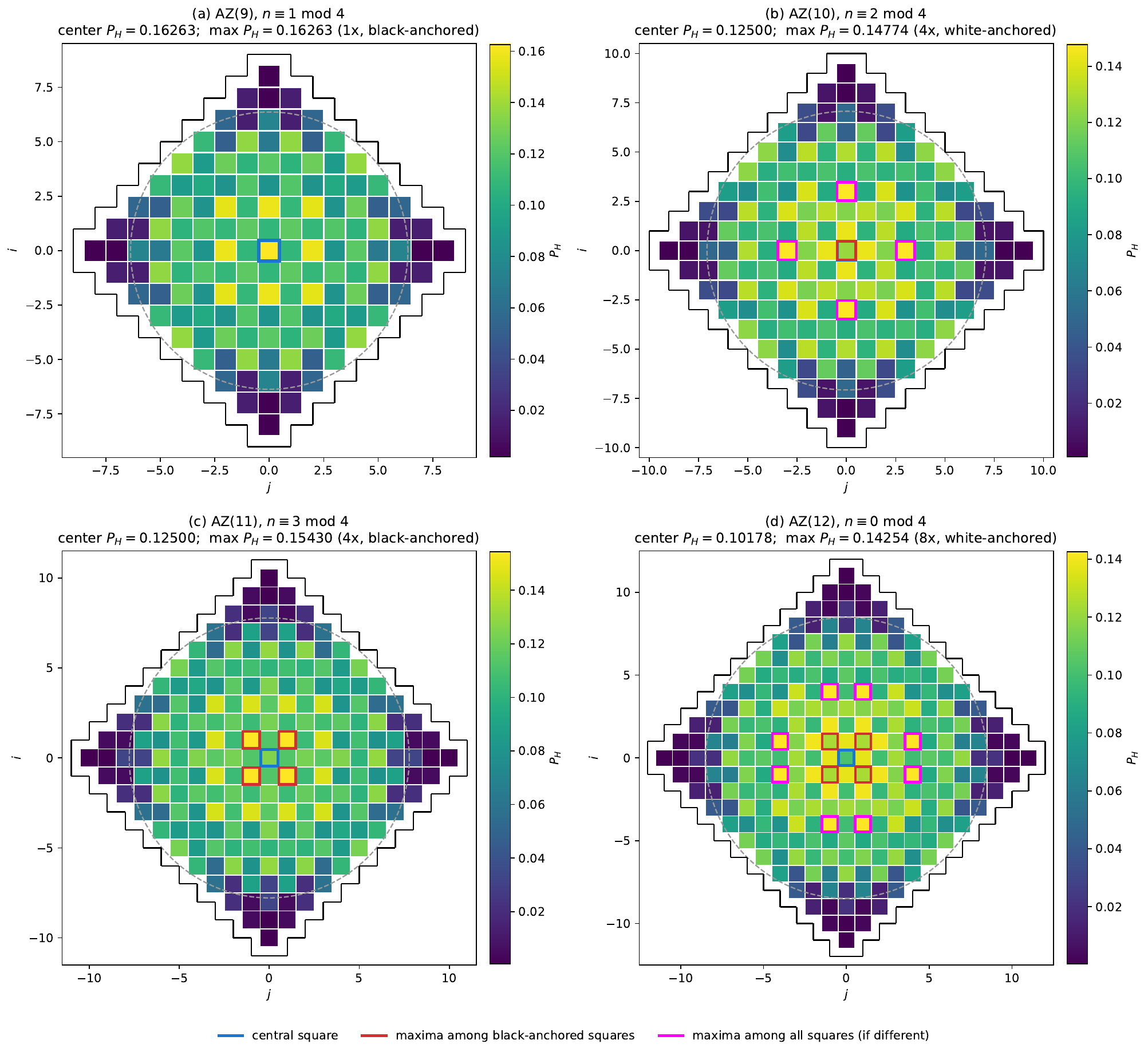}
  \caption{$P_H$ over all elementary squares of (a)~$\Az(9)$, (b)~$\Az(10)$,
    (c)~$\Az(11)$, and (d)~$\Az(12)$, one from each residue class of
    $n\bmod4$, drawn as in Fig.~\ref{fig:spatial-ph}(a).  Blue: central
    square.  Red: maxima among black-anchored squares.  Magenta: global
    maxima, shown when they differ from the black-anchored maxima; in both
    even-$n$ panels they are white-anchored.  Each panel has its own color
    scale.}
  \label{fig:residue}
\end{figure*}

The parity rule decides where $P_H$ is largest (Table~\ref{tab:maxloc}).  Within
the black-anchored class the maximizer follows a mod-4 pattern: the central
square for $n\equiv1,2\pmod4$, and the four diagonal neighbors of the center
for $n\equiv0,3\pmod4$.  For odd $n$ the black-anchored sublattice dominates and
these squares are also the global maxima: a single central peak for $n\equiv1$
and four near-central peaks for $n\equiv3\pmod4$.  For even $n$ the
white-anchored sublattice dominates and the global maximum moves away from the
center.  For $n\equiv0\pmod4$ it consists of eight squares, represented by
$(-5,-2)$ at distance $\sqrt{17}$, for $8\le n\le16$.  For $n\equiv2\pmod4$ it
consists of four squares, represented by $(-4,-1)$ at distance $3$, for
$6\le n\le14$.  The $n\equiv2$ maximizers lie on the symmetry axes and the
$n\equiv0$ maximizers off them, which accounts for the multiplicities $4$ and
$8$.

\begin{table*}[t]
\caption{Maxima of $P_H$ for $4\le n\le17$, separately over black-anchored and
white-anchored squares.  Each maximizer is given by its lexicographically
smallest corner cell $(a,b)$ and the number of symmetry-related squares
attaining the maximum.  The global maximum is shown in bold.}
\label{tab:maxloc}
\begin{ruledtabular}
\begin{tabular}{ccccccc}
 & & \multicolumn{2}{c}{Black-anchored} & \multicolumn{2}{c}{White-anchored} & \\
\cline{3-4}\cline{5-6}
$n$ & $n\bmod4$ & Maximizer & $\max P_H$ & Maximizer & $\max P_H$ & Central $P_H$\\
\hline
4  & 0 & $(-2,-2)\times4$ & 0.117187500 & $(-3,-2)\times8$ & $\mathbf{0.187500000}$ & 0.070312500\\
5  & 1 & $(-1,-1)\times1$ & $\mathbf{0.195312500}$ & $(-3,-2)\times8$ & 0.097656250 & 0.195312500\\
6  & 2 & $(-1,-1)\times1$ & 0.125000000 & $(-4,-1)\times4$ & $\mathbf{0.169921875}$ & 0.125000000\\
7  & 3 & $(-2,-2)\times4$ & $\mathbf{0.171875000}$ & $(-4,-1)\times4$ & 0.108764648 & 0.125000000\\
8  & 0 & $(-2,-2)\times4$ & 0.122528076 & $(-5,-2)\times8$ & $\mathbf{0.149047852}$ & 0.092315674\\
9  & 1 & $(-1,-1)\times1$ & $\mathbf{0.162628174}$ & $(-3,-2)\times8$ & 0.109161377 & 0.162628174\\
10 & 2 & $(-1,-1)\times1$ & 0.125000000 & $(-4,-1)\times4$ & $\mathbf{0.147735596}$ & 0.125000000\\
11 & 3 & $(-2,-2)\times4$ & $\mathbf{0.154296875}$ & $(-4,-1)\times4$ & 0.114692688 & 0.125000000\\
12 & 0 & $(-2,-2)\times4$ & 0.123807907 & $(-5,-2)\times8$ & $\mathbf{0.142543793}$ & 0.101778030\\
13 & 1 & $(-1,-1)\times1$ & $\mathbf{0.150606155}$ & $(-5,-2)\times8$ & 0.114678562 & 0.150606155\\
14 & 2 & $(-1,-1)\times1$ & 0.125000000 & $(-4,-1)\times4$ & $\mathbf{0.139826655}$ & 0.125000000\\
15 & 3 & $(-2,-2)\times4$ & $\mathbf{0.146362305}$ & $(-4,-1)\times4$ & 0.116980148 & 0.125000000\\
16 & 0 & $(-2,-2)\times4$ & 0.124301217 & $(-5,-2)\times8$ & $\mathbf{0.137760954}$ & 0.107006766\\
17 & 1 & $(-1,-1)\times1$ & $\mathbf{0.144390800}$ & $(-5,-2)\times8$ & 0.117416861 & 0.144390800\\
\end{tabular}
\end{ruledtabular}
\end{table*}

The location of the maximum is a more delicate quantity than the sign of
$\Delta(n)$.  For even $n$ the landscape near the center is flat: the three
largest values lie within about $4\%$ of each other at $n=12$ and within
$1.5\%$ at $n=20$.  Additional computations show that the $n\equiv2\pmod4$
maximizer moves further out, to distance $5$ at $n=18$ and $22$, while the
$n\equiv0\pmod4$ maximizer stays at $(-5,-2)$ up to $n=24$.  The class of the
global maximizer followed the parity of $n$ in every case computed
($4\le n\le17$, and even $n$ up to $24$).  The black-anchored maximizer follows the same
mod-4 pattern up to $n=22$.  We do not state a theorem:
Eq.~\eqref{eq:schoenfelder-form} gives the form of $P_H$ at each position, but
locating its maximum would require control of the rational functions $g_1$ and
$g_2$ over the whole diamond.  As $n\to\infty$ the central values approach the
bulk value $1/8$, consistent with the asymptotic local
statistics~\cite{CEP96}.

\subsection{Comparison with the square chessboard}\label{subsec:chess}

Because Eq.~\eqref{eq:local} holds on any planar region, a natural control is
the $L\times L$ chessboard: the same lattice and the same local rule, but a
flat boundary and no frozen region.  We use the phases~\eqref{eq:K} on the cells
$0\le i,j\le L-1$ and write $m=\lfloor L/2\rfloor$ for the half-width, the
analog of $n$.

\textit{Even boards.}---For $L=2m$ the central square $\{m-1,m\}^2$ is
black-anchored and fourfold symmetry again gives $P_H=2P_e(L)^2$.  Exact
enumeration gives $P_e(2)=\tfrac12$, $P_e(4)=\tfrac16$,
$P_e(6)=\tfrac{19}{58}$, $P_e(8)=\tfrac{373}{1802}$,
$P_e(10)=\tfrac{52861}{179786}$, and $P_e(12)=\tfrac{2535297}{11517430}$.  None
of the values we computed (every even $L\le60$ and selected $L$ up to $128$)
equals $1/4$.  The sign of $P_e(L)-\tfrac14$ is $(-1)^{m+1}$, positive for
$L\equiv2$ and negative for $L\equiv0\pmod4$, and its magnitude decays as $1/L$
[Fig.~\ref{fig:chess}(a)].  Numerically, $L(P_e-\tfrac14)\approx+0.42$ and
$-0.40$ at $L\approx60$, and the two branches approach a common magnitude of
about $0.41$.

\textit{Sublattice contrast.}---The contrast $\Delta$, defined as above with
distances measured from the center of the board, obeys the same
rule~\eqref{eq:parity-rule} with $n$ replaced by $m$ (Table~\ref{tab:contrast}).
At equal half-width it is, however, roughly two to six times weaker than in
$\Az(n)$.

\begin{table}[t]
\caption{Sublattice contrast $\Delta$ (mean $P_H$ over black-anchored squares
within distance $3$ of the center minus the mean over white-anchored squares)
for the Aztec diamond $\Az(m)$ and the $2m\times2m$ chessboard.}
\label{tab:contrast}
\begin{ruledtabular}
\begin{tabular}{ccc}
$m$ & $\Az(m)$ & $2m\times2m$ board\\
\hline
4  & $-0.0642$ & $-0.0220$\\
5  & $+0.0845$ & $+0.0319$\\
6  & $-0.0447$ & $-0.0099$\\
7  & $+0.0339$ & $+0.0181$\\
8  & $-0.0420$ & $-0.0073$\\
9  & $+0.0440$ & $+0.0126$\\
10 & $-0.0264$ & $-0.0060$\\
11 & $+0.0217$ & $+0.0097$\\
12 & $-0.0288$ & $-0.0051$\\
13 & $+0.0300$ & $+0.0078$\\
14 & $-0.0187$ & $-0.0045$\\
15 & $+0.0162$ & $+0.0066$\\
\end{tabular}
\end{ruledtabular}
\end{table}

\textit{Spatial profile.}---Figures~\ref{fig:chess}(b) and \ref{fig:chess}(c)
show $P_H$ on the $20\times20$ and $22\times22$ boards.  The profile is
dominated by the boundary.  For every $4\le L\le30$ the global maximum occurs
on the four corner squares, with $P_H\to0.3024$.  Away from the corners $P_H$
oscillates with period two lattice spacings, and the oscillation decays from
the boundary toward the center.  Whether the center falls on a crest or a
trough is decided by the parity of $m$, which is the origin of the sign rule on
the chessboard.

\textit{Odd boards.}---For odd $L$ the board has no domino tiling and its
center is a cell.  The symmetric substitute is the board with its central cell
$(m,m)$ removed, i.e., a single monomer at the center~\cite{FS63}.  Such boards
have been studied from the enumerative side~\cite{Kong06,Tenner05,ByunGoddard25},
and a mobile vacancy in the square dimer model was analyzed in
Ref.~\cite{BBGJ07}.  The phases~\eqref{eq:K} violate the Kasteleyn condition on
the eight-cycle face around the hole.  Flipping the sign of every horizontal
edge that crosses the lattice line $x=m+1$ above the hole restores it; we
checked the resulting edge probabilities against exact enumeration for
$L=3,5,7,9$.  The probability $P_{\mathrm{rad}}$ that the cell directly above
the hole is matched further upward alternates with $L\bmod4$: about $0.43$ for
$L\equiv1$ and $0.39$ for $L\equiv3\pmod4$ at $L\approx40$, with the branches
converging slowly [Fig.~\ref{fig:chess}(a), right axis].  The global maximum
of $P_H$ is again at the corners for every odd $L\le29$.

\begin{figure*}[t]
  \centering
  \includegraphics[width=\textwidth]{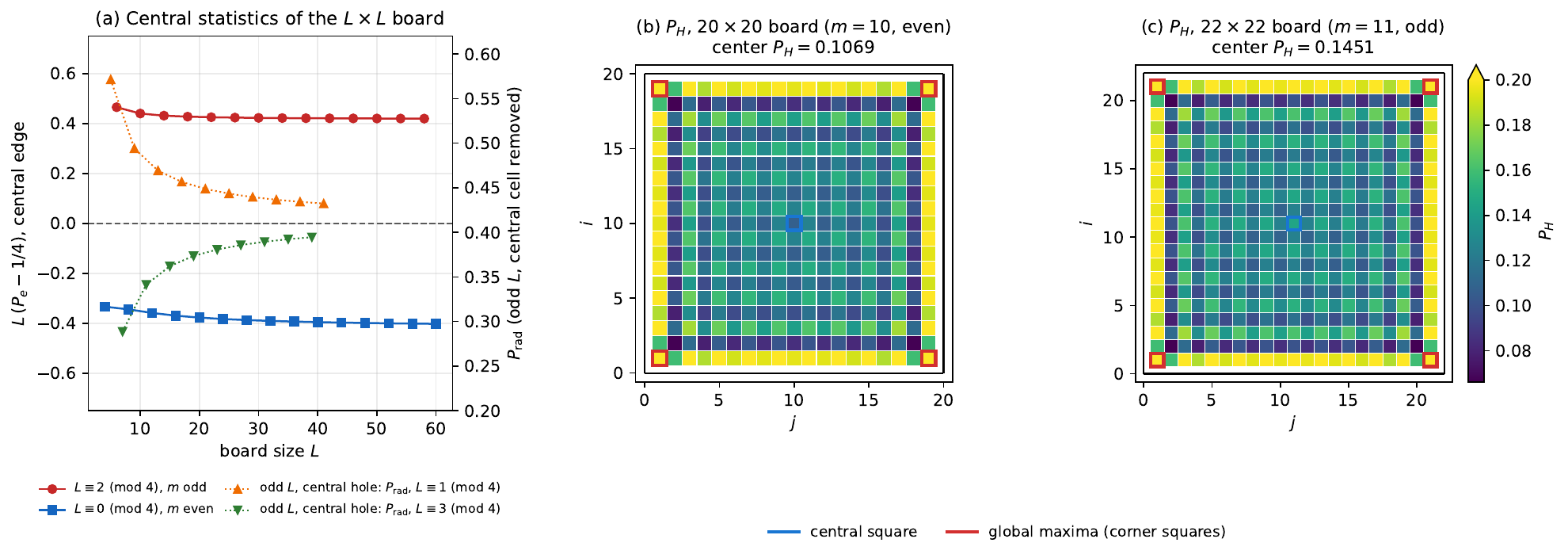}
  \caption{Comparison with the $L\times L$ chessboard.
    (a)~Left axis: scaled deviation $L(P_e-1/4)$ of the central edge
    probability of even boards, grouped by $L\bmod4$.  Right axis: radial edge
    probability $P_{\mathrm{rad}}$ next to the removed central cell of odd
    boards, grouped by $L\bmod4$.
    (b),~(c)~$P_H$ over all elementary squares of the $20\times20$ ($m$ even)
    and $22\times22$ ($m$ odd) boards, drawn as in Fig.~\ref{fig:spatial-ph}(a).
    The color scale is clipped at $0.20$; only the four corner squares exceed
    it (red outlines, $P_H\approx0.302$).  Blue outlines mark the central
    square.}
  \label{fig:chess}
\end{figure*}

\section{Discussion}\label{sec:discussion}

The chessboard comparison shows that the parity rule itself is not specific to the
Aztec geometry.  On both regions the sign is set by the parity of the
half-width, which suggests a boundary-induced even-odd effect carried to the
center by the period-two oscillation of the dimer correlations.  What is
specific to the Aztec diamond is the strength and the consequences of this
effect.  The contrast is several times larger, it is strong enough to move the
global maximum of $P_H$ into the interior of the liquid region, and half of the
residue classes have central values exactly at $1/4$ and $1/8$.  On the
chessboard the contrast is weak, the maxima are pinned to the corners, and no
central value is exactly $1/4$.  The frozen boundary and the Arctic circle are
the natural candidates for this difference.

The mod-4 pattern of the black-anchored maximizer in Table~\ref{tab:maxloc}
has a simple reading in terms of the height function~\cite{Thurston90}.  We fix
$h=0$ at the west corner and let $h$ change by $+1$ along a domino boundary and
by $-3$ across a domino when the cell on the left of the step is black, with the
opposite signs when it is white.  The rotation by $90^\circ$ about the center
maps $\Az(n)$ to itself and exchanges the two colors, so it maps the height
function of each tiling to minus that of the rotated tiling, up to a constant
fixed by the boundary.  Averaging over the uniform measure gives exactly
\begin{equation}\label{eq:mean-height}
\langle h(0,0)\rangle=(-1)^{n+1}n,
\end{equation}
the mean of the west and south corner heights.  Since the height of a vertex is
determined modulo $4$ by its position, the mean central height lies at a
definite offset above this residue: $0$ for $n\equiv0,3$ and $2$ for
$n\equiv1,2\pmod4$.  The four placements of the basketweave tiling, in which the
parallel-domino squares occupy one class of corner cells $(a,b)\bmod2$, are
flat states (Fig.~\ref{fig:plateaus}) whose mean heights, measured from the same residue, are
$0,1,2,3\pmod4$.  The two black-anchored placements take the even values: $2$
for the placement containing the central square and $0$ for the one containing
its diagonal neighbors.  In all four residue classes ($4\le n\le22$) the
black-anchored maximizer belongs to the placement whose height equals the
pinned mean central height.  Because the offset is always even, the mean height
does not account for the sign rule~\eqref{eq:parity-rule}.

Its curvature does.  Let $h_0$ be the mean height at the central vertex of an
elementary square and $h_E,h_W,h_N,h_S$ those at its four neighbors, and let
$x_d=\sigma_d(h_d-h_0)$, where $\sigma_d=\pm1$ is the color of the cell on the
left of the step toward $d$.  The four edges inside the square have
probabilities $(1-x_d)/4$, and Eq.~\eqref{eq:local} becomes exactly
\begin{equation}\label{eq:curvature}
P_H=\frac18-\frac{c\,K}{16}+\frac{x_Nx_S+x_Ex_W}{16},
\end{equation}
where $K=h_E+h_W-h_N-h_S$ and $c=+1$ for black-anchored and $c=-1$ for white-anchored squares.  The
linear term is proportional to the saddle curvature $K$ of the mean height and
has opposite signs on the two sublattices.  It reproduces $\Delta(n)$ to within
$10\%$ for $5\le n\le22$ and on the chessboards of Table~\ref{tab:contrast}, so
the sign rule is a statement about the sign of $K$ near the center.  In the
Aztec diamond this sign has a global origin.  The corner heights are $0$ at the
west and east corners and $(-1)^{n+1}2n$ at the south and north corners, so
the boundary imposes a saddle whose orientation alternates with the parity of
$n$.  The simplest such profile, $\bar h=(-1)^{n+1}\bigl[n+(y^2-x^2)/n\bigr]$,
matches these values and Eq.~\eqref{eq:mean-height} and gives
$\Delta(n)\approx(-1)^{n+1}/(2n)$: the sign of Eq.~\eqref{eq:parity-rule}, the
observed decrease with $n$, and the magnitude to within a factor of two.  On the
chessboard the four corner heights coincide, so there is no such saddle and the
curvature near the center comes only from the period-two oscillation that
decays inward from the boundary (Sec.~\ref{subsec:chess}), consistent with the
weaker contrast found there.

\begin{figure}[tbp]
  \centering
  \includegraphics[width=\columnwidth]{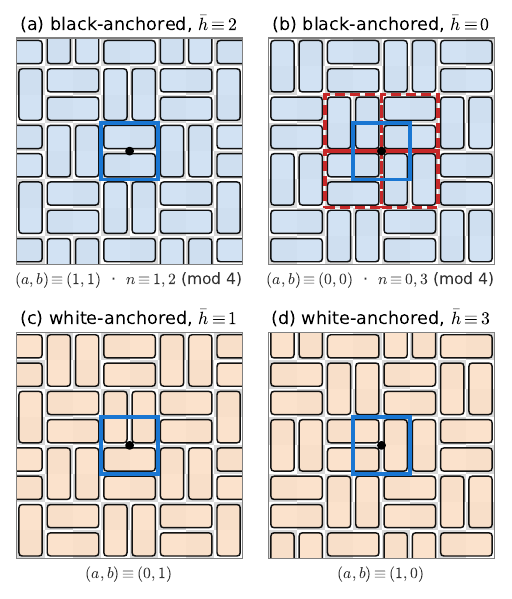}
  \caption{The four placements of the basketweave tiling near the center,
    labeled by the class $(a,b)\bmod2$ of the corner cells of their $2\times2$
    blocks, with their mean height modulo $4$ measured from the residue of the
    central vertex (dot).  Blue: central square.  Red dashed: its four diagonal
    neighbors.  The mean central height of $\Az(n)$, Eq.~\eqref{eq:mean-height},
    lies on plateau (a) for $n\equiv1,2$ and on plateau (b) for $n\equiv0,3\pmod4$;
    these are the placements that contain the black-anchored maximizers of
    Table~\ref{tab:maxloc}.}
  \label{fig:plateaus}
\end{figure}

From the point of view of statistical mechanics, the identity is a constraint
on nearest-neighbor dimer correlations that holds everywhere, in the frozen and
in the liquid region alike.  The sublattice structure is a finite-size
correlation effect of the kind probed in constrained systems such as frustrated
magnets and artificial spin ice, whose low-energy states map onto dimer
coverings~\cite{Henley10,NMS13,Lao18}.  Exact finite-volume data of this type
may serve as benchmarks for such systems.

Several questions remain open.  It would be interesting to derive the parity
rule~\eqref{eq:parity-rule} analytically and to identify the ingredient beyond the mean height that it requires, for instance from the structure of
$g_1$ and $g_2$ in Eq.~\eqref{eq:schoenfelder-form}.  One would also like to
locate the maximum of $P_H$ for large even $n$, to determine the decay of
$\Delta(n)$, and to derive the constant in the $1/L$ decay of the chessboard
central statistics.  Condensation identities for larger faces~\cite{Kuo04}
could extend the present analysis to larger local patterns.

\section{Conclusion}\label{sec:conclusion}

The $2\times2$ correlations of the Aztec diamond have a spatial structure that
the known central results do not reveal.  The elementary squares form two
interleaved sublattices, and the parity of the diamond decides which of them
dominates the center.  For even $n$ this drives the most probable local
arrangement away from the center, and both patterns are read off from the mean
height function: its curvature selects the dominant sublattice, and its pinned
central value the preferred square within it.  The parity rule is shared by the square
chessboard, but its strength and its consequences in the interior are features
of the Aztec geometry.

The shape of a region thus decides not only where its frozen phases lie, but
also which local arrangement is favored deep inside the liquid region.  These
exact finite-size results offer a reference for constrained systems in which
such plaquettes can be imaged directly, such as artificial spin ice, and a
starting point for an analytic theory of the parity rule.

\begin{acknowledgments}
This research was supported by the National Research Foundation (NRF) of Korea
funded by the Korean Government (NRF-2023R1A2C1006050).
The authors used Claude (Anthropic) to assist with writing code for the
numerical computations and figures and with editing the manuscript; all
results were checked by the authors, who take full responsibility for the
content.
\end{acknowledgments}

\section*{Data availability}
The data that support the findings of this article, and the Python scripts that
generate all figures and tables (including an exact-enumeration check of
Eq.~\eqref{eq:local} on regions other than the Aztec diamond), are openly
available at \url{https://github.com/cywon-commits/2x2_superlattice_domino_tiling}.

\bibliography{aztec_pre}

\end{document}